\documentclass[10pt,twoside,leqno,twocolumn]{article}
\usepackage{ltexpprt}

\usepackage{amsmath}
\usepackage{csvsimple}
\usepackage{ifthen}
\usepackage{cite}
\usepackage{enumitem}
\usepackage{subfigure}
\usepackage{graphicx}
\usepackage{todonotes}
\usepackage{url}
\let\MYoriglatexcaption\caption
\renewcommand{\caption}[2][\relax]{\MYoriglatexcaption[#2]{#2}}
\usepackage{thmtools}
\usepackage{thm-restate}

\newtheorem{thm}{Theorem}[section]
\newtheorem{lem}[thm]{Lemma}
\newtheorem{cor}[thm]{Corollary}
\newtheorem{conj}[thm]{Conjecture}
\renewcommand{\theequation}{\arabic{equation}}

\renewcommand{\theenumi}{\roman{enumi}}

\newcommand{\MH}{\operatorname{MH}}

\newcommand{\cuttable}[2][]{
    \ifthenelse{\equal{#1}{}}%
		{}%
		{#1}%
}

\newcommand{\newparentheses}[3]{%
  \expandafter\newcommand\csname #1\endcsname[1]{#2##1#3}%
  \expandafter\newcommand\csname #1L\endcsname[1]{\bigl#2##1\bigr#3}%
  \expandafter\newcommand\csname #1XL\endcsname[1]{\Bigl#2##1\Bigr#3}%
  \expandafter\newcommand\csname #1XXL\endcsname[1]{\biggl#2##1\biggr#3}%
  \expandafter\newcommand\csname #1V\endcsname[1]{\left#2##1\right#3}}

\newparentheses{parens}{(}{)}
\newparentheses{floor}{\lfloor}{\rfloor}
\newparentheses{ceil}{\lceil}{\rceil}
\newparentheses{abs}{|}{|}
\newparentheses{set}{\{}{\}}
\newparentheses{size}{|}{|}

\makeatletter
\newcommand{\onenewattribute}[3]{%
  \@ifundefined{#1}{\let\@@def\newcommand}{\let\@@def\renewcommand}%
  \expandafter\@@def\csname #1\endcsname[2][]{%
    \ifthenelse{\equal{##1}{}}%
    {#2\csname #3\endcsname{##2}}%
    {#2_{##1}\csname #3\endcsname{##2}}}}
\newcommand{\newattribute}[2]{%
  \onenewattribute{#1}{#2}{parens}%
  \onenewattribute{#1L}{#2}{parensL}%
  \onenewattribute{#1XL}{#2}{parensXL}%
  \onenewattribute{#1V}{#2}{parensV}}
\makeatother

\newattribute{OhOf}{\mathrm{O}}
\newattribute{ThetaOf}{\Theta}
\newattribute{OmegaOf}{\Omega}
\newattribute{ohOf}{\mathrm{o}}
\newattribute{omegaOf}{\omega}
\newattribute{lca}{\text{LCA}}

\newattribute{dspr}{d_{\mathrm{SPR}}}
\newattribute{dnni}{d_{\mathrm{NNI}}}
\newparentheses{degree}{|N(}{)|}

\newcommand{\curvature}[2][]{%
    \ifthenelse{\equal{#1}{}}%
		{\kappa(#2)}%
		{\kappa_{#1}(#2)}%
}
\newcommand{\ric}[1]{
	\operatorname{ric}(#1)%
}

\setlist{noitemsep}

\newcommand{\overlap}{%
	\gamma
}

\begin{document}

\title{Ricci-Ollivier Curvature of the Rooted Phylogenetic Subtree-Prune-Regraft Graph\thanks{This work was funded by National Science Foundation award 1223057. Chris Whidden is a Simons Foundation Fellow of the Life Sciences Research Foundation.}}

\author{Chris Whidden\thanks{Program in Computational Biology, Fred Hutchinson Cancer Research Center, Seattle, WA, USA 98109. \set{cwhidden,matsen}@fredhutch.org} \\
\and
Frederick A Matsen IV\footnotemark[2]}
\date{}

\maketitle

\begin{abstract}
Statistical phylogenetic inference methods use tree rearrangement operations such as subtree-prune-regraft (SPR) to perform Markov chain Monte Carlo (MCMC) across tree topologies.
These methods are known to mix quickly when sampling from the simple uniform distribution of trees but may become stuck in the local optima of multi-modal posterior distributions for real data induced by non-uniform likelihoods.
The structure of the graph induced by tree rearrangement operations is an important determinant of the mixing properties of MCMC, motivating study of the underlying \emph{rSPR graph} in greater~detail.

In this paper, we investigate the rSPR graph in a new way: by calculating Ricci-Ollivier curvature with respect to uniform and Metropolis-Hastings random walks.
We confirm using simulation that mean access time distributions depend on distance, degree, and curvature, showing the relevance of these curvature results to stochastic tree search.
These calculations require fast new algorithms for constructing and sampling these graphs, reducing the time required to compute an rSPR graph from $O(m^2n)$-time to $O(mn^3)$, where $m$ is the (often large) number of trees in the graph and $n$ their number of leaves, and reducing the time required to select an SPR neighbor of a tree uniformly at random to $O(n)$ time.
We then develop a closed form solution to characterize how the number of SPR neighbors of a tree changes after an SPR operation is applied to that tree.
This gives bounds on the curvature, as well as a flatness-in-the-limit theorem indicating that paths of small topology changes are easy to traverse.
However, we find that large topology changes (i.e. moving a large subtree) gives pairs of trees with negative curvature.
Although these pairs of trees with negative curvature do not impede mixing in this simple well-connected space, they may manifest as bottlenecks in the much smaller credible sets induced by phylogenetic posteriors with a likelihood function.
This work extends our knowledge of the rSPR graph, in particular properties that are relevant for investigation of sampling the rSPR graph.

\end{abstract}

\section{Introduction}
Molecular phylogenetic methods reconstruct evolutionary trees from DNA or RNA data and are of fundamental importance to modern biology.
Statistical phylogenetics is the currently most popular means of reconstructing phylogenetic trees, in which the tree is viewed as an unknown parameter in a likelihood-based statistical inference problem.
The likelihood function in this setting is the likelihood of generating the observed sequences via a continuous time Markov chain (CTMC) evolving down the tree starting from a sequence assumed to be sampled from the stationary distribution~\cite{felsenstein1981evolutionary}.
The lengths of the branches of the phylogenetic tree give the ``time'' parameter in the CTMC, where the generated sequence accrues mutations, typically in an IID manner across sites.
It is now common for researchers to approximate the posterior distribution of trees and their associated parameters in a Bayesian setting using Markov chain Monte Carlo (MCMC).

In order to estimate these distributions accurately, MCMC samplers must sufficiently explore the set of trees.
Phylogenetic search algorithms typically attempt to do so through a combination of modifications to the continuous parameters and tree topology.
Topology changes have been identified as the main limiting factor of Bayesian MCMC algorithms~\cite{lakner2008efficiency,hohna2012guided}, as other parameters cannot be accurately estimated if the topology distribution is not accurately sampled.
Commonly used phylogenetics software packages such as MrBayes~\cite{Ronquist2012-hi} and BEAST~\cite{bouckaert2014beast} rearrange subtrees via subtree-prune-regraft (SPR) moves (Figure~\ref{fig:spr}) or the subset of SPR moves called nearest neighbor interchanges (NNI)~\cite{robinson1971comparison}.
Thus, phylogenetic searches can be viewed as traversing the \emph{SPR graph}: the graph with phylogenetic trees as vertices and SPR adjacencies as edges.

It has become increasingly clear that the structure of the SPR graph plays an important role in determining the accuracy of tree searches.
Researchers have previously identified slow mixing in MCMC with pathological data~\cite{Mossel2005-ly,Mossel2006-fo,Ronquist2006-fv}.
On the other hand, fast mixing has been identified with exceptionally well-behaved data~\cite{Stefankovic2011-hu} or with a uniform distribution~\cite{spade2014note}.
Studies on real data~\cite{beiko2006searching, lakner2008efficiency}, however, have identified posteriors which are difficult to sample using MCMC.
Previously, the lack of sufficient computational tools for examining phylogenetic posteriors in terms of SPR operations made it difficult to determine the cause of these difficulties.
By developing the first such tools, we recently showed that graph structure has a significant effect on MCMC mixing with MrBayes applied to real data~\cite{Whidden2015-yi}, and that multimodal posteriors are common and separated by ``bottlenecks" of specific classes of SPR moves.

Although the SPR graph is thus very important in determining the success of phylogenetic inference procedures, still little is known about the rooted or unrooted versions of the SPR graph itself.
\cite{Song2003-gf}~developed a recursive procedure on a tree to find the degree of the corresponding vertex in the rooted SPR (rSPR) graph, and corresponding bounds on degree.
\cite{Ding2011-bj}~showed that the diameter $\Delta_{\text{rSPR}}$ of the rSPR graph is $n - \Theta(\sqrt n)$, and for the unrooted case they show
\begin{equation}
n - 2\ceil{\sqrt{n}} + 1
\le \Delta_{\text{uSPR}}(n)
\le n - 3 - \floorV{\frac{\sqrt{n - 2} - 1}{2}}.
\end{equation}
We are not aware of any further work investigating properties of the SPR graph, which may be due to its complexity.
Indeed, even computing the distance between topologies in terms of SPR operations (rooted and unrooted) is NP-hard~\cite{bordewich05,hickey2008sdc}.
Fortunately, it is fixed-parameter tractable with respect to the distance in the rooted case~\cite{bordewich05} and efficient fixed-parameter algorithms have recently been developed~\cite{whidden2013hybridization,Whidden2015-yi} and begun to allow such investigation.

Ollivier and colleagues recently pioneered a new approach to calculating Ricci curvature on a general type of metric space, including graphs~\cite{Ollivier2009-bw,Joulin2010-jg}.
In this framework, local information about the metric space is given by a random walk (rather than a Riemann tensor) such that their notion of curvature formalizes the notion of to what extent random walking brings points together.
Applying the framework to Brownian motion on a manifold returns the classical definition of Ricci curvature.
Curvature is determined by the ratio of the earth mover's distance~\cite{rubner2000earth} between neighborhoods of a pair of vertices given by a random walk and the distance between the vertices.
Here the term \emph{random walk} on a space $X$ simply denotes a family of probability measures parameterized by points of $X$ satisfying reasonable assumptions, which includes biased walks such as MCMC.
This approach has been useful for determining properties of a wide variety of graphs including the internet topology~\cite{ni2015ricci} and cancer networks~\cite{sandhu2015graph}.

In this paper, we investigate curvature of the rSPR graph with respect to two random walks and compare those results to access times (i.e. hitting times) for those random walks.
Our explicit focus here is to investigate random walks defined only in terms of the graph itself: the uniform random walk and MCMC sampling from the uniform prior on trees.
In future work, we will extend these methods to study more complicated distributions with non-uniform topology probabilities.

We required several new computational tools.
We present a fast new algorithm for computing rSPR graphs from a set of trees, reducing the time to do so from $\OhOf{m^2n}$ to $\OhOf{mn^3}$ for a set of $m$ trees with $n$ leaves.
As the full rSPR graph on trees with $n$ leaves contains $(2n-3)!! = 3 \cdot 5 \cdot \ldots \cdot (2n-3)$ trees, this is a significant improvement in practice for exploring large subsets of the graph (or, as we do here, the full graph for small numbers of leaves).
By exploiting symmetries in the rSPR graph, we were able to calculate all of the curvatures for pairs of trees with up to seven leaves.
By carefully examining the overlap in rSPR moves, we present a new method for computing the degree of a tree in the rSPR graph that allows one to select an rSPR neighbor uniformly at random in linear-time without explicitly generating the graph.
This stands in contrast to the sampling methods used in current software such as MrBayes, which do not propose SPR moves uniformly.

\begin{figure*}[t]
	\hspace*{\stretch{1}}
	\subfigure[\label{fig:x-tree}]{\includegraphics[scale=1.25]{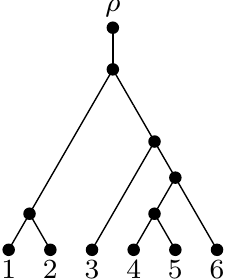}}
	\hspace*{\stretch{2}}
	\subfigure[\label{fig:subtree}]{\includegraphics[scale=1.25]{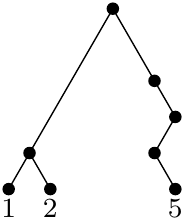}}
	\hspace*{\stretch{2}}
	\subfigure[\label{fig:induced}]{\includegraphics[scale=1.25]{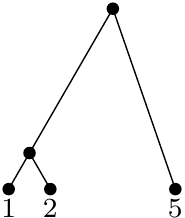}}
	\hspace*{\stretch{1}}
	\\
	\hspace*{\stretch{1}}
	\subfigure[\label{fig:spr}]{\includegraphics[scale=1.25]{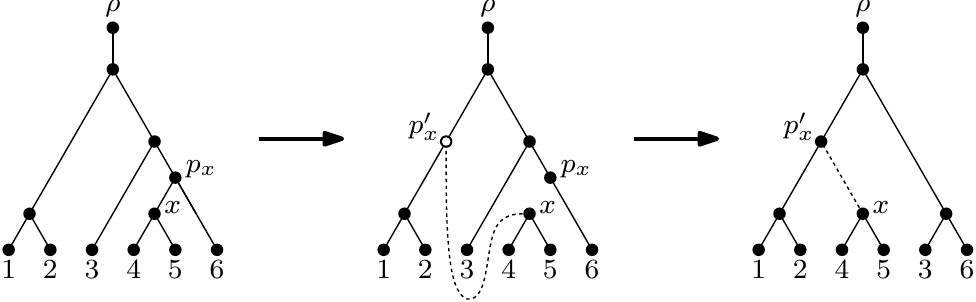}}
	\hspace*{\stretch{1}}
	\\
	\hspace*{\stretch{1}}
	\hspace*{\stretch{1}}

	\caption{(a) An $X$-tree $T$.
		(b) $T(V)$, where $V = \set{1,2,5}$.
		(c) $T|V$.
		(d) An rSPR operation transforms $T$ into a new tree $T'$ by \emph{pruning} a \emph{subtree} and \emph{regrafting} it in another location.}
	\label{fig:trees}
\end{figure*}

Using our methods to simulate these random walks, we found that the distribution of access times between pairs of trees can be described by distance between the trees, the degrees of the trees, and the curvature.
Moreover, we found that rSPR graphs for trees with 7 or more leaves have tree pairs with negative curvature, corresponding to direct paths that are difficult to traverse stochastically.
By getting a more fine-tuned understanding of the rSPR neighborhood of pairs of vertices, we are able to give bounds on the earth mover's distance in this context and thus curvatures under these random walks.
In particular, we present a full characterization of the change in rSPR degree that occurs from a given rSPR move and find that even though they each count as one move, rSPR moves which modify large subtrees are less likely to be explored during these random walks.
Pairs of trees separated by such moves correspond to the pairs with negative curvature identified in our simulation results.
These pairs occur infrequently in these well-connected graphs, however, they may be more problematic in real posterior distributions where the majority of probability is spread over a relatively small number of trees~\cite{Whidden2015-yi}.
In summary, we extend knowledge about an important graph for phylogenetics, specifically in a way that models phylogenetic MCMC search.

The automated computational analysis code can be found at \url{https://github.com/matsengrp/curvature}.
Proofs of our theorems and lemmas can be found in the appendix.

\section{Preliminaries}

\begin{figure*}
	\hspace*{\stretch{1}}
	\includegraphics[scale=1.25]{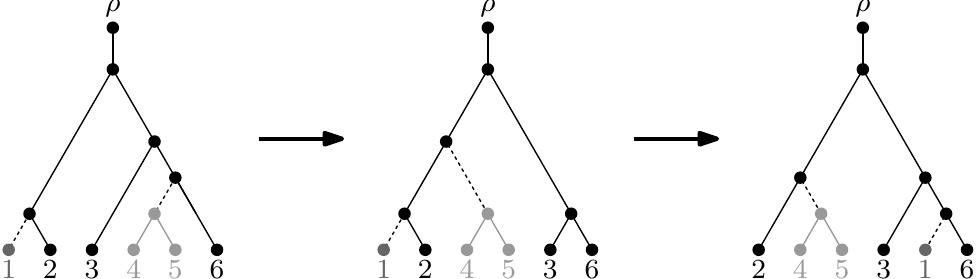}
	\hspace*{\stretch{1}}
	\caption{\
    Two rSPR operations, each of which moves one grey subtree.
    The leftmost and rightmost trees are rSPR distance two apart.}
	\label{fig:two-spr}
\end{figure*}

We follow the definitions and notation from~\cite{bordewich05,whidden2013hybridization, Whidden2015-yi}.
A (rooted binary phylogenetic) $X$-tree is a rooted tree $T$ whose nodes have zero or two children such that the leaves of $T$ are bijectively labelled with the members of a label set $X$.
As in~\cite{bordewich05,whidden2013hybridization,Whidden2015-yi}, the tree is augmented with a labelled root node $\rho$ and $\rho$ is considered a member of $X$ (Fig.~\ref{fig:x-tree}).
We generally use $n$ to refer to the number of leaves in an $X$-tree.
For a subset $V$ of $X$, $T(V)$ is the smallest subtree of $T$ that connects all nodes in $V$ (Fig.~\ref{fig:subtree}).
The $V$-tree induced by $T$ is the smallest tree $T|V$ that can be obtained from $T(V)$ by suppressing unlabelled nodes with fewer than two children (Fig.~\ref{fig:induced}).
For the rest of the paper, \textbf{we will assume that all phylogenetic trees are binary and rooted}, and that tree inclusion is rooted tree inclusion.

A \emph{parent (sub)tree} of a subtree $U$ is the smallest subtree strictly containing $U$.
A \emph{parent edge} of a subtree $U$ is the edge connecting $U$ to the rest of the tree.
The \emph{internal edges} of a tree are the edges that do not contact a leaf or $\rho$.
A \emph{ladder tree} (also known as a \emph{caterpillar tree}) is a tree such that every internal node has a leaf as a direct descendant.
A \emph{balanced tree} is a tree such that the sum of the depths of internal nodes is minimum over all trees with the same number of leaves.
The \emph{least common ancestor} (LCA) of a set $R$ of two or more nodes is the unique node that is an ancestor of each node $r \in R$ and at maximum depth.
Similarly, the LCA of two or more subtrees is the LCA of their parent nodes.

A (rooted) \emph{subtree-prune-regraft} (rSPR) operation on an $X$-tree $T$ cuts an edge $e = (x, p_x)$ where $p_x$ denotes the parent of node $x$.
$T$ is divided into two subtrees $T_x$ and $T_{p_x}$ containing $x$ and $p_x$, respectively.
Then the operation adds a new node $p'_x$ to $T_{p_x}$ by subdividing an edge of $T_{p_x}$ and adding a new edge $(x, p'_x)$, making $x$ a child of $p'_x$.
Finally, $p_x$ is suppressed, joining the two edges on either side of that node.
See Figure~\ref{fig:spr} for an example.
The inclusion of $\rho$ allows for rSPR moves which move subtrees to the root of the tree.

rSPR operations give rise to a distance measure between $X$-trees: $\dspr{T_1, T_2}$ is the minimum number of rSPR operations required to transform an $X$-tree $T_1$ into $T_2$.
For example, the trees in Figure~\ref{fig:two-spr} are separated by two rSPR operations.
Moreover, rSPR operations naturally give rise to a graph on the set of $X$-trees for which this distance is simply the shortest-path graph distance.
Let $\mathcal{T}_n$ be the set of trees with $n$ leaves and label set $X = \set{1, 2, \ldots n, \rho}$.
Then the rSPR graph $G$ of $\mathcal{T}_n$ is the graph with vertex set $V(G) = \mathcal{T}_n$ and edge set $E(G) = \set{(T,S) \mid \dspr{T,S} = 1, T \in V, S \in V}$.

To avoid confusion between the two types of graph structures considered here, we refer to vertices of the rSPR graph as \emph{vertices} and vertices of individual trees (i.e.\ leaves and internal nodes) as \emph{nodes}.
Let $N(T)$ be the set of rSPR neighbors of a tree $T$ (this does not include $T$).
For example, the tree $T$ with 4 leaves in Figure~\ref{fig:neighborhood} has 10 neighbors.
We say that the degree of $T$ is $\degree{T}$, that is, the number of trees which can be obtained from $T$ by a single rSPR operation.
We assume that all trees are bifurcating, and thus use degree to refer only to the degree of rSPR graph vertices.

Ricci-Ollivier curvature provides a rigorous yet intuitive formalization of the shape of a metric space with respect to a random walk.
For the purposes of this paper, we will specialize to that space being a graph equipped with the shortest-path distance.
For a more rigorous presentation in the more general setting of a Polish metric space, see~\cite{Ollivier2009-bw} or the survey~\cite{Ollivier2010-ao}.

Let $m_x$ and $m_y$ be probability densities of the position of a specified random walk after one step of the random walk, starting at points $x$ and $y$ of a graph $G = (V,E)$, respectively.
The transportation distance~\cite{Villani2003-wv} (equivalently Wasserstein distance, or ``earth movers distance''~\cite{rubner2000earth}) between $m_x$ and $m_y$ is the minimum amount of ``work'' required to move $m_x$ to $m_y$ along edges of the graph, that is
\vspace{-0.5em}
\begin{equation}
W_1(m_x, m_y) := \min_{\xi \in \Pi(m_x, m_y)} \sum_{\{z,w\} \subset V} d(z,w) \xi(z,w),
\end{equation}
where $d(z,w)$ is the graph shortest-path distance ($\dspr{z,w}$ in our case) and $\Pi(m_x, m_y)$ is the set of densities on $V \times V$ that are $m_x$ after projecting on the first component and $m_y$ after projecting on the second.

The \emph{coarse Ricci-Ollivier curvature} of $x$ and $y$ is then defined as:
\vspace{-0.5em}
\begin{equation}
\curvature{m; x, y} := 1 - \frac{W_1(m_x, m_y)}{d(x, y)}.
\label{eq:curvatureDef}
\end{equation}
For the purposes of this paper, ``curvature'' without further specification will refer to \eqref{eq:curvatureDef}.
We will use $\curvature{x,y}$ to denote the curvature of the simple (uniform choice of neighbor) random walk, and use $\curvature{\MH; x,y}$ to indicate curvature with respect to the Metropolis-Hastings random walk sampling the uniform distribution (described in detail in Section~\ref{sec:random_walks}).
Positive curvature implies that the neighborhoods $m_x$ and $m_y$ are closer in transportation distance than point masses at $x$ and $y$, zero curvature implies that they are neither closer nor farther, and negative curvature implies that $m_x$ and $m_y$ are more distant than point masses at $x$ and $y$.
Curvature thus provides an intuitive measure of the difficulty of moving between regions of the graph with a random walk.
\begin{figure}[!h]
	\hspace*{\stretch{1}}
	\includegraphics[width=0.43\textwidth]{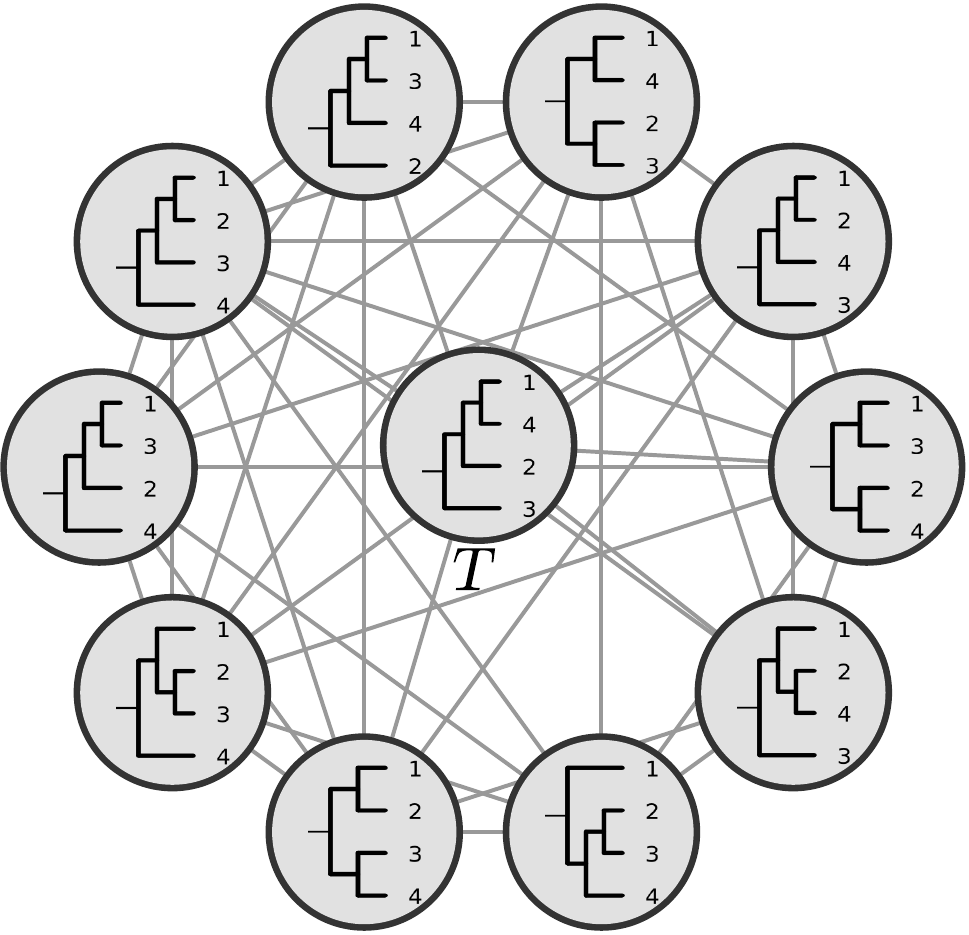}
	\hspace*{\stretch{1}}
	\caption{The neighborhood of an $X$-tree $T$ with 4 leaves, showing connections between neighbors.}
	\label{fig:neighborhood}
\end{figure}

Lin et al.~\cite{lin2011ricci} defined a variant definition of curvature in terms of lazy random walks which Loisel and Romon~\cite{Loisel2014-hu} dubbed the \emph{asymptotic Ricci-Olivier curvature}.
The lazy random walk only travels according to $m_x$ with probability $p$ and otherwise stays put.
Thus the lazy mass assignment $m^p_x$ is the sum of $p\, m_x$ and a point mass of $1 - p$ on $x$.
We denote the coarse curvature of the $p$-lazy random walk between two vertices $x$ and $y$ with respect to a random walk $m$ by $\curvature[p]{m; x,y}$.
For example, $\curvature[1/4]{m; x,y}$ describes the curvature of the lazy random walk that follows the given random walk $m$ with probability $1/4$ and remains stationary with probability $3/4$.
The asymptotic Ricci-Ollivier curvature of $x$ and $y$ is then:
\begin{equation}
\ric{m; x, y} := \lim_{p \rightarrow 0} \frac{\curvature[p]{m; x,y}}{p} .
\label{eq:asympCurvatureDef}
\end{equation}
As above for $\kappa$, we use $\ric{x, y}$ as shorthand for $\ric{m; x, y}$ when $m$ is the uniform lazy random walk, and $\ric{\MH; x, y}$ when $m$ is the Metropolis-Hastings random walk sampling the uniform distribution (Section~\ref{sec:random_walks}).
This definition of curvature is invariant of $p$ for small $p$~\cite{Loisel2014-hu} and can be used to avoid parity problems on graphs where the uniform random walk is periodic without choosing a specific laziness parameter (e.g. Ollivier often considered $\curvature[\frac{1}{2}]{x,y}$ for this purpose).
As we prove in Lemma~\ref{lem:asymptotic}, the notions of coarse and asymptotic curvature differ only by a small factor bounded by $\frac{2}{\max(\degree{x}, \degree{y})}$ between adjacent vertices and are equal for nonadjacent vertices.

\section{Efficient algorithms for computing and sampling rSPR graphs}

\subsection{Computing the rSPR graph of $m$ trees with $n$ leaves in $\OhOf{mn^3}$-time.}
\label{sec:computing_treespace}

It is necessary to have an efficient method of constructing the full rSPR graph for a fixed number of leaves in order to study it.
The previous best algorithm for this problem requires $\OhOf{m^2 n}$ time, where $m$ is the number of trees in the graph and $n$ the number of leaves~\cite{Whidden2015-yi}.
Here we reduce that time to $\OhOf{mn^3}$.
Note that for the full rSPR graph, $m$ is the rapidly growing function $(2n-3)!!$, that is, $3 \cdot 5 \cdot \ldots \cdot (2n-3)$, and this is therefore a significant improvement in practice, as we demonstrate below.

In previous work~\cite{Whidden2015-yi}, we constructed (unrooted) SPR graphs from subsets of $m$ high probability trees sampled from phylogenetic posteriors to compare mixing and identify local maxima.
Although the SPR distance (rooted and unrooted) is NP-hard to compute~\cite{bordewich05,hickey2008sdc}, it is fixed-parameter tractable with respect to the distance in the rooted case~\cite{bordewich05}.
In particular, one can determine in $\OhOf{n}$-time whether two rooted phylogenetic trees are adjacent in the rSPR graph ($\OhOf{n^2}$-time for unrooted trees) using the algorithms of Whidden et al.~\cite{whidden2009unifying,whidden2010fast, whidden2013hybridization,Whidden2015-yi}.
We applied this method comparing each of the $m$ trees pairwise to identify adjacencies, requiring a total of $\OhOf{m^2n}$-time ($\OhOf{m^2 n^2}$-time in the unrooted case).
However, this method is impractical when applied to construct graphs with 7 or more leaves, due to the rapidly growing $\OhOf{m^2}$ factor.

The key to our efficient algorithm for quickly computing dense rSPR graphs (those containing a significant portion of the full rSPR graph) lies in avoiding the pairwise comparison of non-adjacent trees and thereby shaving off an $\OhOf{m}$ factor.
The input to our algorithm is a set $\mathcal{T}$ of phylogenetic trees in the $\OhOf{n}$-length Newick \cite{wiki:newick} representation of each tree as a string.
These representations are made unique by ordering each tree so that leftmost subtrees contain the smallest alphanumeric label of descendants.
We construct a mapping from each tree $T_i$ to its order index in this list $i$.
Begin with an empty graph $G$.
For each tree $T_i$, we first add a vertex $i$ to the graph and then use Corollary~\ref{cor:enumerate_neighbors} below to enumerate the $\OhOf{n^2}$ neighbors of $T_i$ in the rSPR graph in $\OhOf{n^3}$-time.
This efficient enumeration procedure is the key step required to achieve our desired running time of $\OhOf{mn^3}$.
We use the tree to index mappings to determine whether these trees are already vertices of the graph and, if so, add an edge in the graph from $T_i$ to each such neighbor $T_j$.
The high-level steps are as follows, and we show in Theorem~\ref{thm:construct_graph} that this algorithm is correct and can be implemented to run in the stated time.

\vspace{1em}
\textsc{Construct-rSPR-Graph($\mathcal{T}$)}
\begin{enumerate}[label={\arabic*}.]
	\item Let $G$ be an empty graph.
	\item Let $M$ be a mapping from trees to integers.
	\item Let $i = 0$.
	\item For each of the $m$ trees: \vspace{-0.2em}
		\begin{enumerate}
			\item Add a vertex $i$ to $G$ representing the current tree $T_i$.
			\item Add $T_i \rightarrow i$ to $M$.
			\item For each of the $\OhOf{n^2}$ neighbors of $T_i$, enumerated using \textsc{Enumerate-rSPR-Neighbors($T_i$)}:
				\begin{enumerate}
					\item If the current neighbor $T_j$ is in $M$ then add an edge $(i,M[T_j])$ to $G$.
				\end{enumerate}
		\item $i = i + 1$.
		\end{enumerate}
\end{enumerate}

\begin{restatable}{thm}{constructgraph}
	\label{thm:construct_graph}
	The subgraph of the rSPR graph induced by a set $\mathcal{T}$ of $m$ trees with $n$ leaves can be constructed in $\OhOf{mn^3}$-time.
\end{restatable}

We implemented this procedure in the C++ program \texttt{dense\_spr\_graph} of the software package \texttt{spr\_neighbors}~\cite{spr_neighbors}, which outputs an edge list format graph suitable for input to other software.
The construction procedure reduced the time required to compute the 10,395-vertex 7-taxon rSPR graph from 2,104.68 seconds to 12.71 seconds on an Intel Core 2 Duo E7500 desktop running Ubuntu 14.04.
Moreover, although we do not study the 135,135-vertex 8-taxon rSPR graph in this paper, our algorithm required only 303.45 seconds to construct it on the same hardware.
Constructing the 8-taxon rSPR graph using the previous method required 377,395 seconds (more than 4 days), and thus that method is infeasible for constructing larger rSPR tree graphs.
Thus, we believe our fast graph construction procedure will itself be useful for further studies of rSPR graph subsets similar to~\cite{Whidden2015-yi}, as the algorithm can quickly construct rSPR graphs for any given subset of trees.

\subsection{Simulating random walks on the rSPR graph.}
\label{sec:random_walks}
The uniform random walk moves from one vertex to one of its neighbors uniformly at random, which makes this walk more likely to sample higher degree vertices.
In contrast, the Metropolis Hastings (MH) random walk with constant likelihood function proposes a move from a tree $T$ to a neighbor tree $S$ uniformly at random and then accepts the move according to the Hastings ratio, $\min\left(1, \frac{\degree{T}}{\degree{S}}\right)$.
The MH random walk is guaranteed to sample each tree uniformly at random and is therefore representative of a phylogenetic MCMC program sampling trees under a uniform prior.

To efficiently simulate the MH random walk, we developed a linear-time algorithm for proposing rSPR moves that does not require the rSPR graph to be explicitly built and stored in memory.
A na\"ive approach would require $\OhOf{n^3}$ time: $\OhOf{n}$ time to generate each of the $\OhOf{n^2}$ neighbors of a given tree so that one could be picked uniformly at random.
To eliminate an $\OhOf{n^2}$ factor, we developed a deterministic ordering of rSPR moves with a one-to-one correspondence to rSPR neighbors, as described in the next paragraph.
Given such an order, a uniform neighbor can be selected by its index in $\OhOf{n}$ time.
We note that the recursive formula of Song~\cite{Song2003-gf} for the degree of a tree does not group rSPR moves that move a particular subtree, and thus would still require $\OhOf{n^2}$ time to select a specific rSPR neighbor by index.

We consider the distribution of rSPR moves in terms of the number of nodes contained within a subtree.
Recall that a tree with $n$ leaves has $2n-1$ total nodes (ignoring the artificial $\rho$ node).
Given a subtree $R$ with $x$ nodes, observe that there are $2n-1 - x$ possible locations to regraft $R$.
However, some of these moves will result in the same neighboring tree as other rSPR moves.
In particular, where we call the edge connecting the subtree rooted at that node to the rest of the tree the ``node's edge'', we have:
\begin{enumerate}
\item Moving $R$ to its sibling edge results in the same tree, not a neighboring tree,
\item Moving $R$ to its parent edge results in the same tree,
\item Moving $R$ to its grandparent edge is the same as moving its aunt to its sibling edge, and
\item Moving $R$ to its aunt edge is the same as moving its aunt to $R$'s edge.
\end{enumerate}
We prove in Lemma~\ref{lem:compute_degree} that this list is exhaustive, that is each other pair of $R$ and destination edge $e$ results in a unique rSPR neighbor.
We assign $(2n-1-x)-2$ moves to children of the original non-$\rho$ root (lacking both an aunt and a grandparent), and $(2n-1-x)-4$ moves to each other non-root node.
Let $N(T,u)$ denote the neighbors of $T$ assigned to node $u$, obtained by moving the subtree $R$ rooted at $u$.
We thus achieve a new method for computing the neighborhood size:

\begin{restatable}{lem}{computedegree}
	\label{lem:compute_degree}
	For a tree $T$ with $n$ leaves,
	$$\degree{T} = \sum_{u \in T} \size{N(T,u)},$$
	for nodes $u$ of $T$, where $N(T,u)$ is as defined above, and:
	$$\size{N(T,u)} = \begin{cases}
		2n - x - 5 &\text{if depth($u$) $ > 1$, } \\
		2n - x - 3 &\text{if depth($u$) $ = 1$ } \\
		0 &\text{if depth($u$) $ \le 0$ } \\
	\end{cases}.$$
\end{restatable}

In particular, this formulation implies a total ordering of rSPR moves such that every move moving the same subtree $R$ forms a contiguous subsequence.
We can thus apply the following algorithm to select a neighbor uniformly at random for a tree $T$:

\vspace{1em}
\textsc{Select-rSPR-Neighbor($T$)}
\begin{enumerate}[label={\arabic*}.]
	\item	Compute the degree of $T$, $\degree{T}$ using Lemma~\ref{lem:compute_degree}.
	\item Pick a random integer $r$ in the range $[1,\degree{T}]$.
	\item Label each node $u$ of $T$ by its preorder number and compute the number of nodes in the subtree rooted at each $u$.
\item For each tree node $u$ and while $r > 0$:
	\begin{enumerate}
		\item Decrease $r$ by $\size{N(T,u)}$.
		\item If $r < 0$, let $S$ be the $\abs{r}$ member of $N(T,u)$ and terminate the for loop.
	\end{enumerate}
\item Return the neighbor $S$.
\end{enumerate}

\begin{figure*}
	\hspace*{\stretch{1}}
	\subfigure[6 taxa]{\includegraphics[width=3in]{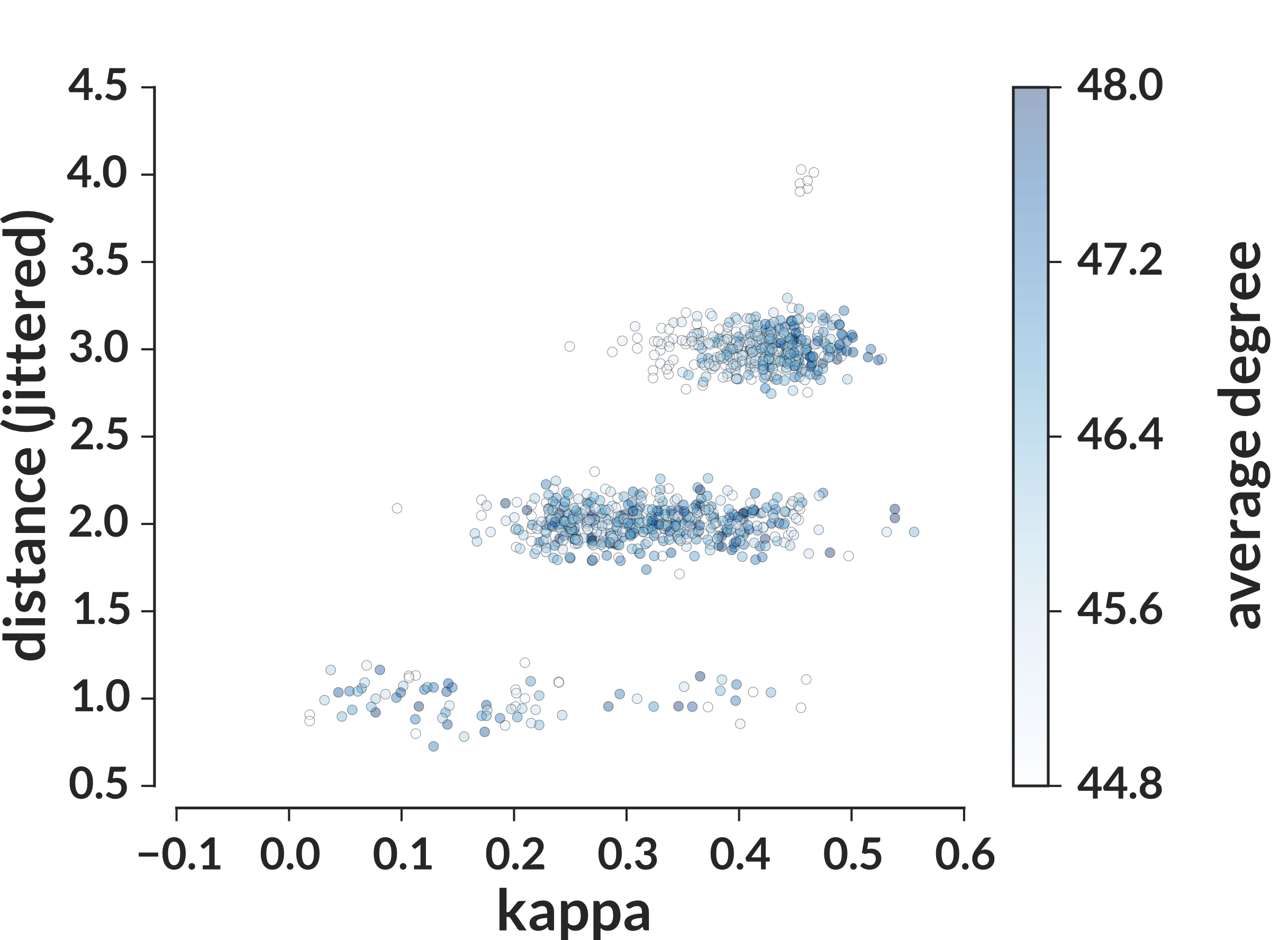}}
	\hspace*{\stretch{2}}
	\subfigure[7 taxa]{\includegraphics[width=3in]{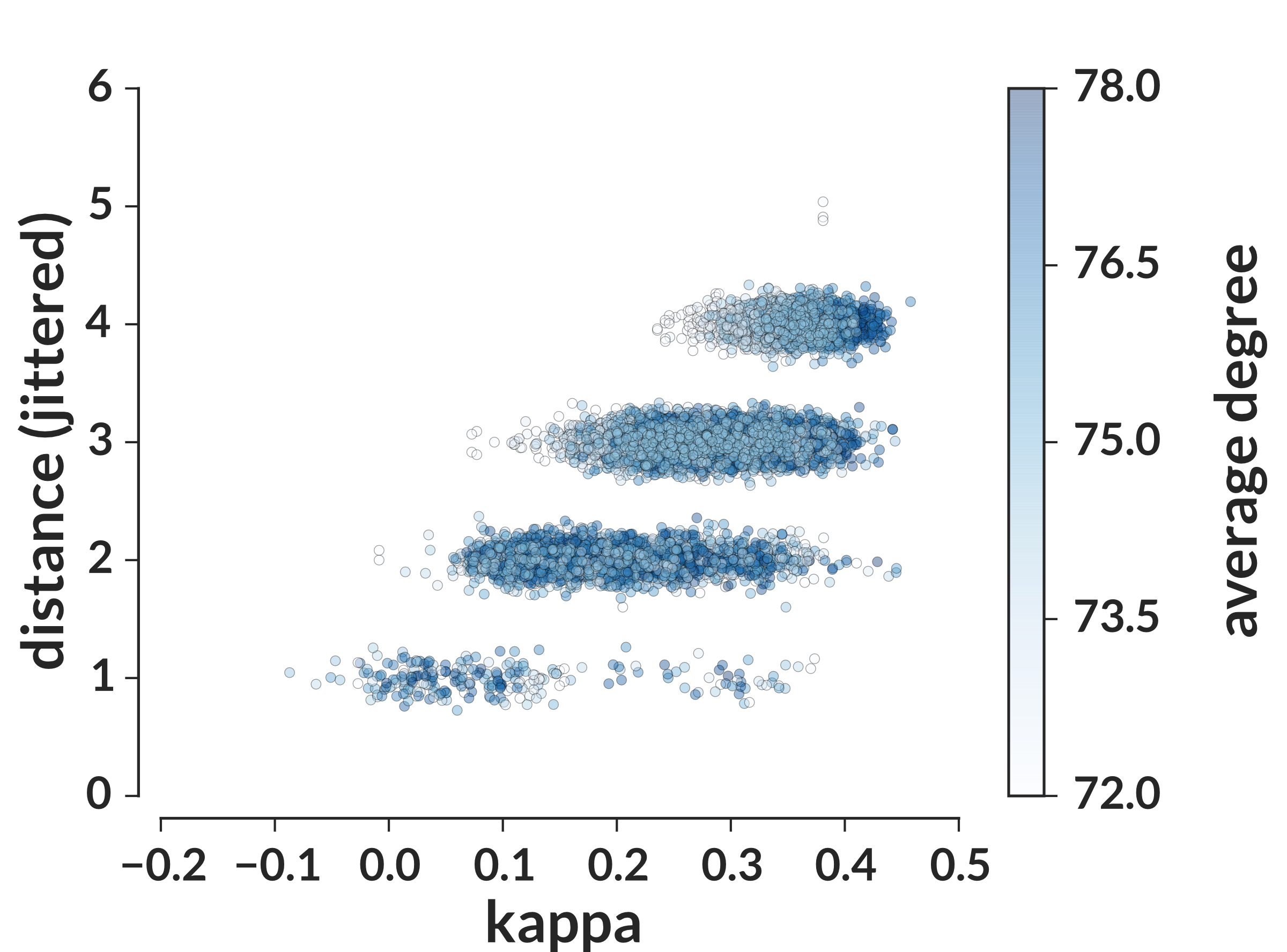}}
	\hspace*{\stretch{1}}
    \caption{Scatter plot of $\curvature{\MH; T_1,T_2}$ values versus $\dspr{T_1,T_2}$ for the rSPR graph. Color displays the average degree of $T_1$ and $T_2$. Distance values randomly perturbed (``jittered'') a small amount to avoid superimposed points.}
	\label{fig:rspr-scatter}
\end{figure*}

\begin{restatable}{lem}{selectrandomneighbor}
	\label{lem:select_random_neighbor}
	An rSPR neighbor of a tree $T$ can be chosen uniformly at random in $\OhOf{n}$-time using $\OhOf{n}$ space.
\end{restatable}

Observe that this procedure can be easily adapted to explore the full neighborhood of a tree in $\OhOf{n^3}$ time, which we use for Theorem~\ref{thm:construct_graph}. We call the resulting procedure \textsc{Enumerate-rSPR-Neighbors($T$)}.
We thus have the following corollary:

\begin{cor}
\label{cor:enumerate_neighbors}
	The rSPR neighbors of a tree $T$ can be enumerated in $\OhOf{n^3}$-time.
\end{cor}

We implemented this procedure in the C++ package \texttt{random\_spr\_walk}~\cite{random_spr_walk}.
We sampled a 200,000-iteration random walk on the 4-leaf rSPR graph and a 50,000-iteration random walk on the 5-leaf rSPR graph.

\section{Access times of random walks on the rSPR graph can be understood using distance, degree, and curvature}

\subsection{Computing curvature values.}
To compute curvature values, we first used \texttt{dense\_spr\_graph} to compute the rSPR graph for four to seven leaves, as discussed in Section~\ref{sec:computing_treespace}.
We then computed curvatures for given pairs of trees directly, by using linear programming~\cite{Loisel2014-hu} to compute the minimal mass transport $W_1$ using the SAGE \cite{SAGE} front-end to the GLPK~\cite{glpk} solver; code can be found in \cite{gricci} which grew from the code described in \cite{Loisel2014-hu}.

This would have required an enormous amount of computation to directly compute curvatures for the $((2n-3)!!)^2$ pairs of trees with $n$ leaves, even for the small values of $n$ we consider here.
We instead exploited the fact that pairs of trees which are equivalent modulo label renumbering are symmetric in the rSPR graph and therefore guaranteed to have the same curvature.
For example, the pairs
$\{(((1,2),3),4), ((1,2),(3,4))\}$ and
$\{(((1,4),2),3), ((1,4),(2,3))\}$
are the same after relabeling, so their curvatures are the same.
We thus directly computed curvature values for one representative pair from each such equivalence class, or \emph{tanglegram} \cite{Venkatachalam2010-zh}; the group-theoretic enumeration methods are described in a manuscript in preparation, and the SAGE \cite{SAGE} and GAP4 \cite{GAP4} code is at \cite{tangle}.

We find a wide variation in curvature among tanglegrams (Figure~\ref{fig:rspr-scatter}).
Curvature values tended to increase with increasing rSPR distance, and their variance decreased with increasing distance.
Neighboring trees achieved minimum curvature values for a given number of leaves, and we found maximum curvature values between trees at maximum distance or one rSPR move closer than the maximum.
This suggests that the increased difficulty of moving between trees with a random walk due to distance may be tempered somewhat by larger curvature in the highly connected rSPR~graph.

Larger rSPR graphs tended to have smaller curvature values.
Indeed, the 7-leaf rSPR graph contained adjacent pairs of trees with negative curvature.
Such pairs indicate difficult paths for phylogenetic searches, which may be exacerbated by likelihood or branch length constraints.

\begin{figure*}
	\hspace*{\stretch{1}}
	\subfigure[\label{fig:short-time-kappa-access}]{\includegraphics[width=3in]{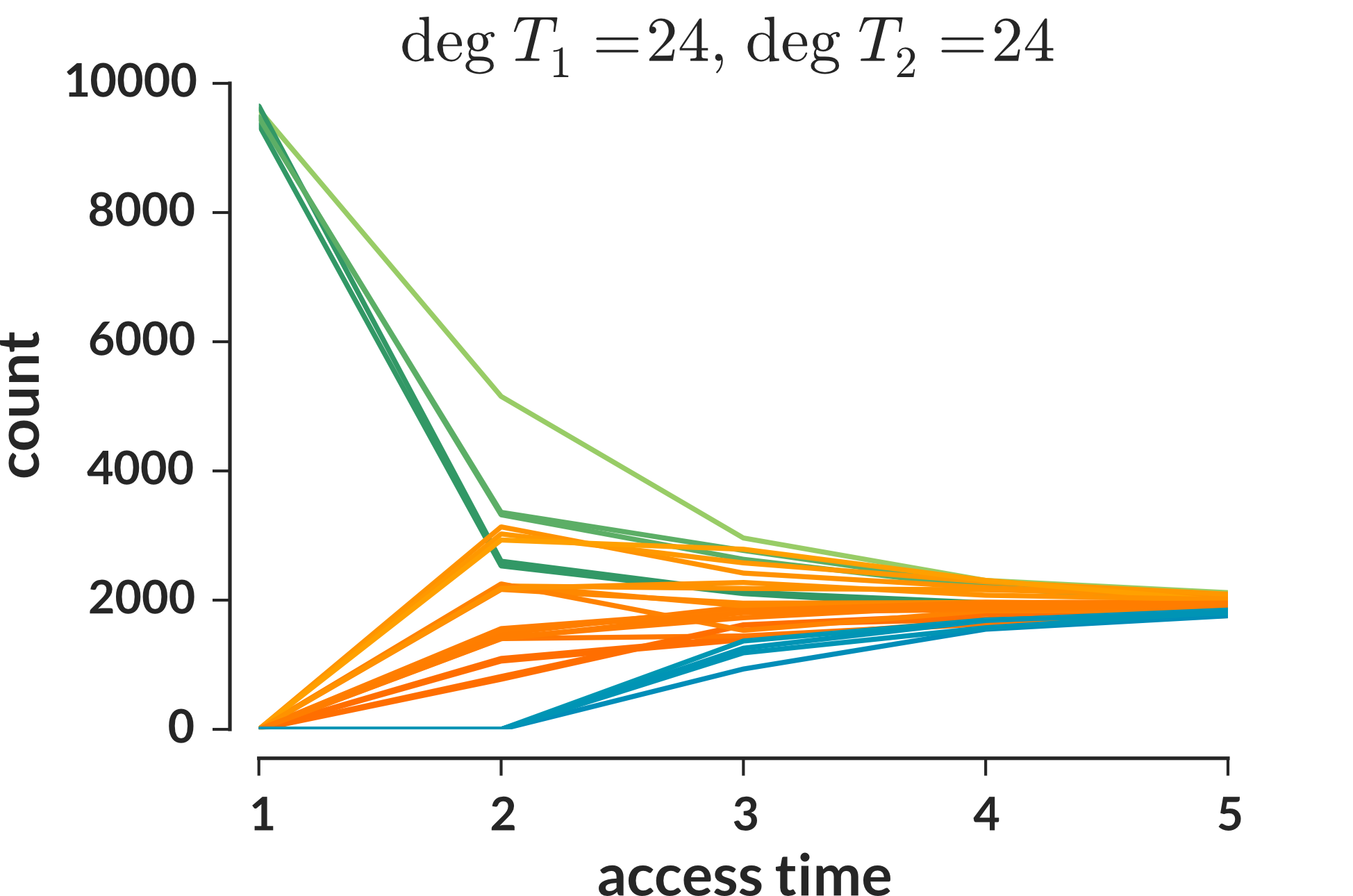}}
	\hspace*{\stretch{2}}
	\subfigure[\label{fig:long-time-kappa-access}]{\includegraphics[width=3in]{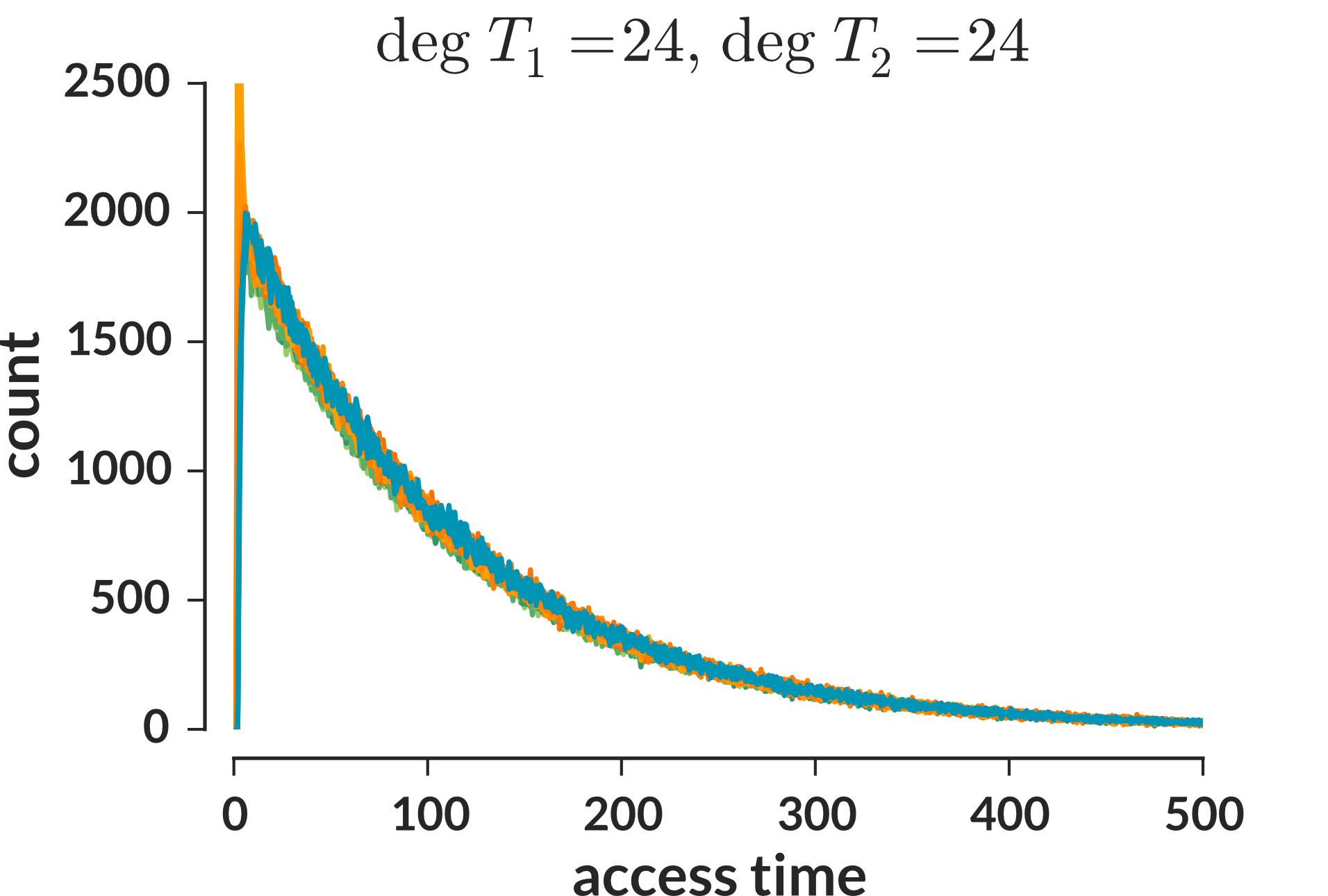}}
	\hspace*{\stretch{1}}
    \caption{\
        Distribution of rSPR $\MH$ access times for those pairs of 5-taxon trees with degree 24 that are not simple inclusions of 4-taxon pairs of trees.
        Color signifies rSPR distance between the trees, with green, orange, and blue signifying distances of 1, 2, and 3, respectively; the saturation of the color shows coarse curvature $\curvature{\MH; \cdot, \cdot}$, such that increased saturation (i.e.\ darker color) indicates a smaller $\kappa$.
        }
	\label{fig:kappa-access}
\end{figure*}

\subsection{Access time simulation.}
The access time for a pair of vertices in a graph is the (random) number of iterations required to go from one of the vertices to the other in a random walk~\cite{lovasz1993random}; we were interested in the connection between curvature and access time.
In previous work, we computed mean access times (MAT) between pairs of trees in  MCMC random walks: the mean number of iterations required to move from one tree to the other.
We applied this work to demonstrate the influence of SPR graph structure on real MCMC posteriors sampled with MrBayes~\cite{Whidden2015-yi} using \texttt{sprspace}~\cite{sprspace}.

\begin{table}
\centering
\caption{p-values for ordinary least squares linear multiple regression of rSPR mean access time against degree and distance (two-tailed $t$-test of regression coefficient). The p-values for 7 taxa are smaller than the machine precision used to calculate them.}
\begin{filecontents*}{tables/mean-access.csv}
variable,5 taxa,6 taxa,7 taxa
$T_1$ degree,2.425e-07,2.726e-55,0
$T_2$ degree,0.04367,4.302e-21,0
$d_{\operatorname{rSPR}}$,5.026e-09,1.104e-44,0
\end{filecontents*}
\csvautotabular{tables/mean-access.csv}
\label{tab:regressionMAT}
\end{table}

Here, to gain more insight, we used simulation to approximate the entire access time distribution.
Again we use the insight that the access time for a pair of trees with a simple random walk does not depend on the actual labeling of those trees, but rather only on their relative labeling.
Thus rather than enumerate access times between trees, which would have required a tremendous amount of memory and computational power to obtain accurate estimates, we enumerate times between pairs of trees in a tanglegram.
To calculate the empirical distributions of access times we aggregate all access times for the same tanglegram using our group-theoretic methods \cite{tangle}.

\begin{table}
\centering
\caption{p-values for ordinary least squares linear multiple regression of rSPR $\delta_1$ against degree, distance, and $\kappa$ (two-tailed $t$-test of regression coefficient).}
\begin{filecontents*}{tables/delta1.csv}
variable,5 taxa,6 taxa,7 taxa
$T_1$ degree,9.376e-05,2.944e-07,5.51e-09
$T_2$ degree,0.2366,0.1432,0.1687
$d_{\operatorname{rSPR}}$,5.151e-06,0.0007557,3.276e-23
$\kappa$(MH),4.462e-06,1.436e-22,1.459e-46
\end{filecontents*}
\csvautotabular{tables/delta1.csv}
\label{tab:regressionDelta}
\vspace{1em}
\end{table}

We find that the mean access time between trees $T_1$ and $T_2$ is determined by $\degree{T_1}$ and $\degree{T_2}$ (Table~\ref{tab:regressionMAT}).
Furthermore, plotting the distribution of access times between pairs of trees with respect to their distance and curvature hints that smaller $\kappa$ slightly shifts the distribution of access times towards larger access times (Fig.~\ref{fig:short-time-kappa-access}).
We quantify this effect by defining $\delta_1$ to be the difference between the first pair of access time counts such that the second entry in the pair is nonzero.
For example, $\delta_1$ for distance 1 pairs (green lines in Fig.~\ref{fig:kappa-access}) is the count for time 1 minus the count for time 2, while $\delta_1$ for distance 3 pairs (blue lines in Fig.~\ref{fig:kappa-access}) is the count for time 2 minus the count for time 3.
Regression finds a clear influence of $\kappa$ on $\delta_1$ (Table~\ref{tab:regressionDelta}).
This confirms the intuitive interpretation of $\kappa(T_1, T_2)$ as quantifying the propensity of a random walk to go from $T_1$ to $T_2$ relatively directly, certainly before the random walk achieves stationarity.
On the other hand, if the random walk starting from $T_1$ does not quickly arrive at $T_2$ and instead achieves stationarity, the original position of the random walk is forgotten, and the access time is then a standard exponentially distributed waiting time for an event in a Poisson process (Fig.~\ref{fig:long-time-kappa-access}).

\begin{figure*}
	\hspace*{\stretch{1}}
	\includegraphics[width=0.8\textwidth]{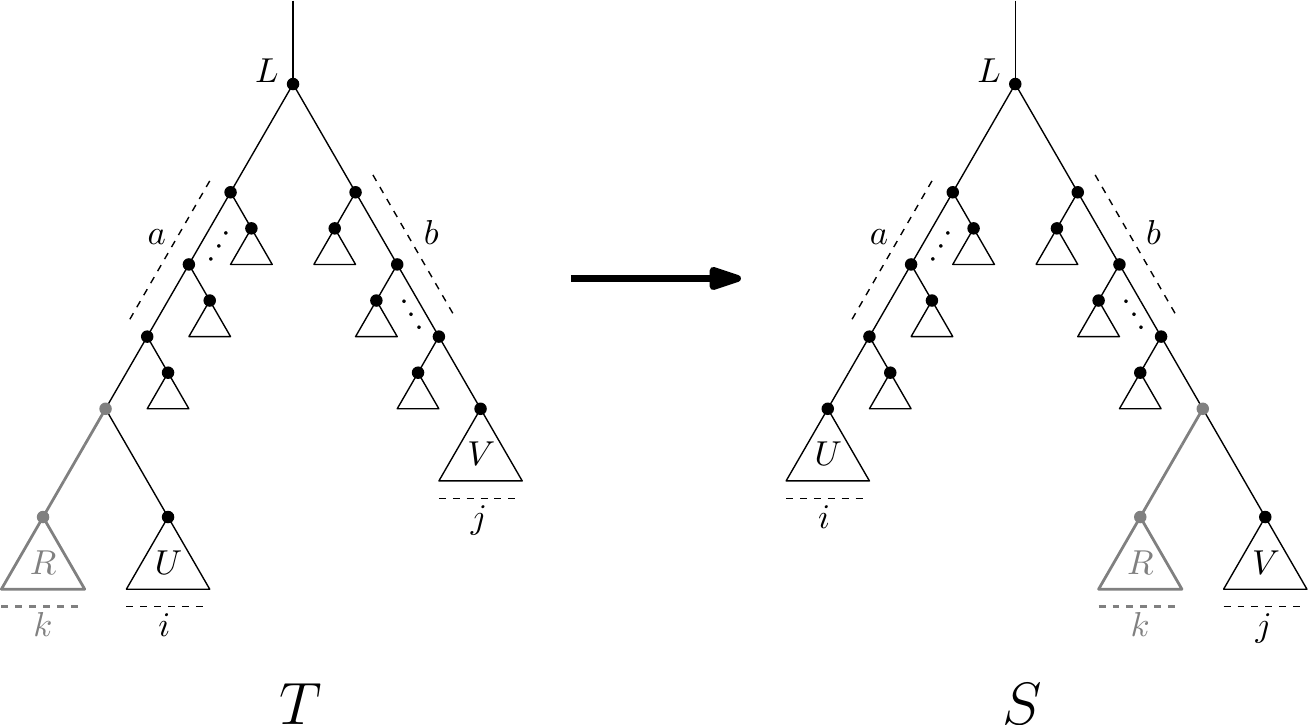}
	\hspace*{\stretch{1}}
	\caption{An rSPR move labelled as in Lemma~\ref{lem:degree_change}.
Moving the grey subtree $R$ from its position adjacent to $U$ in tree $T$ to its position adjacent to $V$ in tree $S$ changes the rSPR degree by $2\left(k(a-b) + i - j\right)$.}
	\label{fig:spr-degree-change}
\end{figure*}

The analysis can be reproduced by invoking the SCons (\url{http://scons.org/}) build tool and running the cells in an IPython notebook; instructions are in the repository README file.

\section{Rooted SPR Neighborhoods}
\label{sec:neighborhoods}
Having made the connection between curvature values and access times on rSPR graphs, we now consider curvature theoretically.
We begin by bounding differences between degrees, and then continue by considering features relevant to the earth mover's distance that we call ``squares" and ``triangles" in the rSPR graph.
Many of our results in this section follow from a characterization of the change in degree and distribution of permissible rSPR moves after an rSPR move is applied.

\begin{lem}[{Song~\cite{Song2003-gf}}]
\cuttable{    \pushQED{\qed}}
	\label{lem:degree_extremes}
	For a tree $T$ with $n$ leaves:
	\begin{enumerate}
		\item $\degree{T} = 3n^2 - 13n + 14$, if $T$ is a ladder tree,
		\item $\degree{T} = 4(n-2)^2 - 2 \sum_{m=1}^{n-2} \floor{\log_2(m+1)}$, if $T$ is a balanced tree, and
		\item  $3n^2 - 13n + 14 \le \degree{T} \le 4(n-2)^2 - 2 \sum_{m=1}^{n-2} \floor{\log_2(m+1)}$, otherwise.
	\end{enumerate}
\cuttable{    \popQED}
\end{lem}

We now bound the ratio and difference of rSPR degree between two trees with $n$ leaves.
\begin{restatable}{lem}{degreemaxdelta}
	\label{lem:degree_max_delta}
	Let $T$,$S$ be trees with $n \ge 3$ leaves, and assume w.l.o.g. that $\degree{T} \le \degree{S}$.
	Then:
	\begin{enumerate}
		\item $\frac{\degree{T}}{\degree{S}} \ge 3/4$, and
		\item $\degree{S} - \degree{T} \le n^2 - 5n + 6$.
	\end{enumerate}
\end{restatable}

We can improve these bounds in the case of adjacent trees.
To do so, we require the following lemma that characterizes how the degree of a tree changes after an rSPR operation.
See Figure~\ref{fig:spr-degree-change} for an illustration.

\begin{restatable}{lem}{degreechange}
\label{lem:degree_change}
Let $T$ and $S$ be trees such that $S$ can be obtained from $T$ by moving a subtree $R$ with $k$ leaves from its position adjacent to subtree $U$ to a location adjacent to subtree $V$.
Let $L$ be the $\lca{U,V}$ in $T$.
Let $a$ be the number of intermediate nodes on the path from the parent of $R$ to $L$ in $T$, excluding endpoints.
Similarly, let $b$ be the number of intermediate nodes on the path from $V$ to $L$ in $T$, excluding endpoints.
Let $i$ be the number of leaves in $U$ and $j$ be the number of leaves in $V$, excluding any leaves of $R$.
Then the degrees of $T$ and $S$ differ by:
$$2\left(k(a-b) + i - j\right).$$

\end{restatable}

Moreover, we can use these ideas to determine the number of rSPR moves that are, in some respects, independent of a given rSPR move.
That is, for two trees $S$ and $T$ differing by a single rSPR move, we wish to know the number of rSPR moves that are applicable to both trees rather than unique to one of the trees.
To formalize this concept, consider pairs of trees $T' \in N(T)$ and $S' \in S(T)$ such that $\dspr{T',S'} = 1$.
The number of such ``squares'' involving two adjacent trees will play a key role in our curvature bounds, as they push the curvature of those trees towards 0.

\begin{restatable}{cor}{pairedneighbors}
\label{cor:paired_neighbors}
Continuing with the setting and notation in Lemma~\ref{lem:degree_change}, at least
$$\overlap := \deg(T) - 2kb - 2(j-1) = \deg(S) - 2ka - 2(i-1)$$
trees in the neighborhood of $T$ can be paired with $o$ trees in the neighborhood of $S$ such that the pairings are disjoint and $\dspr{T',S'} = 1$ for each $(T',S')$ pair.
\end{restatable}

We can now use Lemma~\ref{lem:degree_change} to improve the bounds in Lemma~\ref{lem:degree_max_delta} for two adjacent trees.
\begin{restatable}{lem}{degreemaxdeltaadjacent}
	\label{lem:degree_max_delta_adjacent}
	Let $T$, $S$ be trees with $n \ge 3$ leaves, s.t. $\degree{T} \le \degree{S}$ and $\dspr{T,S} = 1$.
	Then:
	\begin{enumerate}
		\item $\degree{S} - \degree{T} \le 2\floor{\frac{n-2}{2}}\ceil{\frac{n-2}{2}} \le \frac{1}{2} (n-2)^2 $,
		\item $\frac{\degree{T}}{\degree{S}} \ge \frac{5}{6}$, $\forall n \ge 4$, and
		\item $\lim_{n\rightarrow\infty}\frac{\degree{T}}{\degree{S}} =  \frac{6}{7}$.
	\end{enumerate}
\end{restatable}

Next, we bound the number of neighbors shared by two adjacent trees.
The number of such ``triangles'' involving two adjacent trees has a key role in determining whether their curvature is positive or negative.

\begin{restatable}{lem}{sharedneighbors}
	\label{lem:shared_neighbors}
Let $T$ and $S$ be trees such that $\dspr{T,S} = 1$.
Then $\size{N(T) \cap N(S)} \le 6n - 17$.
\end{restatable}

\section{Curvature}
We now consider properties of the uniform (a.k.a.\ isotropic) random walk on the $n$-leaf rSPR graph.
Recall that the uniform random walk begins at a tree $T$ and moves to a tree uniformly at random from $N(T)$.
Recall that the coarse uniform random walk curvature between two trees $T$ and $S$ is $\curvature{T, S} := 1 - \frac{W_1(m_T, m_S)}{d(T, S)}$, where $W_{1,n}$ is the mass transport term \eqref{eq:curvatureDef}.
For the uniform random walk, $m_T$ is the probability measure assigning a mass of $\frac{1}{\degree{T}}$ to each of $T$'s neighbors.
Our results follow from the lemmas of Section~\ref{sec:neighborhoods}.

\begin{restatable}{thm}{flatcurvature}
\label{thm:flat_curvature}
Fix a positive integer $k$ and let $R$ be a tree with $k$ leaves.
Let $\{T_n \mid n > k\}$ be a sequence of trees all containing $R$, and let $\{S_n \mid n > k\}$ be the same sequence $T_n$ but with $R$ cut off and attached at a different location.
Then $\lim_{n \rightarrow \infty} \curvature{T_n,S_n} = 0$ for the uniform random walk on the rSPR graph.
\end{restatable}

Next we note a simple and rough bound on the curvature of two trees with respect to their distance, then obtain a tighter bound on the maximum curvature of two adjacent trees.

\begin{restatable}{lem}{curvaturedistancebound}
	\label{lem:curvature_distance_bound}
	Let $T$ and $S$ be two trees. Then:
	$$ \frac{-2}{\dspr{T,S}} \le \curvature{T,S} \le \frac{2}{\dspr{T,S}} .$$
\end{restatable}

\begin{restatable}{lem}{maxadjacentcurvature}
	\label{lem:max_adjacent_curvature}
	The maximum curvature between two adjacent trees with $n$ leaves is $\frac{6n-17}{3n^2-13n+14}$.
\end{restatable}
This bound is tight and has been verified computationally for $n \le 7$.

It is more difficult to obtain a closer bound on the maximum curvature of nonadjacent trees.
Lemma~\ref{lem:curvature_distance_bound} suggests that more distant pairs of trees should have smaller curvatures than close trees as neighborhood effects decrease with respect to the increasing distance.
However, our experiments with $n \le 7$ suggest that maximum curvature tends to increase with distance (with respect to a fixed $n$), as a far greater fraction of the neighbors approach each other as the distance increases.
Indeed, for $5 \le n \le 7$ the maximum curvature is obtained by pairs of trees at one less than the maximum distance.
Moreover, nearly all of the neighbors of these pairs approach each either.
We thus conjecture the following:
\begin{conj}
	Let $k_n$ be the maximum curvature between two trees with $n$-leaves.
	Then:
	\begin{enumerate}
		\item $k_n \le \frac{2}{\Delta_{\text{rSPR}}(n)-1}$, and
		\item $ k_n \sim \frac{2}{\Delta_{\text{rSPR}}(n)-1}$.
	\end{enumerate}
\end{conj}
Proving or disproving this conjecture would go a long way toward understanding the effects of relative distance on curvature.
However, we suspect that this will require a greater understanding of the distribution of tree neighborhoods with respect to one another than is currently known.
Next, we bound the minimum curvature of two adjacent trees.

\begin{restatable}{lem}{curvaturelowerbound}
\label{lem:curvature_lower_bound}
	The curvature between adjacent trees with $n$ leaves is at least
	$$\frac{-n^2 + 2n}{3.5n^2 - 15n + 16}.$$
\end{restatable}

We further observe that the limit of our curvature lower bound is $-\frac{2}{7}$.
Complete enumeration with $n \le 7$ show that no pair of trees have curvature less than $-\frac{2}{5}$ and our bound meets or exceeds this value for $n > 7$.
Moreover, the rSPR distance is a metric, so this bounds the curvature for arbitrary pairs of trees (Proposition 19 of~\cite{Ollivier2009-bw}).
This directly leads to the following Corollary:

\begin{cor}
	The curvature between two trees is at least $-\frac{2}{5}$.
\end{cor}

Note that this bound is not tight (at least for small $n$) as it is rarely necessary to transport mass the maximum distance between unpaired trees.
We also note that the lower bounds in this section do not follow from the more general setting described in \cite{Jost2013-ce}.
However, the pair of trees used in the proof of Lemma~\ref{lem:curvature_lower_bound} will always have negative curvature, for all $n \ge 7$.

We next bound the difference between the coarse and asymptotic curvatures.
Recall that $\curvature[p]{T,S}$ is the coarse Ricci-Ollivier curvature between trees $T$ and $S$ with respect to the lazy walk that remains at a given tree with probability $1-p$ and moves with probability $p$.
For the lazy uniform random walk, $m_T$ is now $T \cup N(T)$, with each neighbor assigned mass $\frac{p}{\degree{T}}$ and $T$ assigned the remaining $1 - p$ mass.
The asymptotic Ricci-Ollivier curvature $\ric{T,S}$ is $\lim_{p \rightarrow 0} \curvature[p]{T,S} / p$.
As we now prove, these two notions of curvature differ only by a small factor inversely proportional to the maximum degree of $T$ and $S$.

\begin{restatable}{lem}{asymptotic}
\label{lem:asymptotic}
	Let $T$ and $S$ be trees with $n$ leaves.
	Then:
	\begin{enumerate}
		\item	$\ric{T,S} = \curvature{T,S}$, if $\dspr{T,S} > 1$,
		\item	$\curvature{T,S} \le \ric{T,S} \le \curvature{T,S} + \frac{2}{\max(\degree{T}, \degree{S})}$, if $\dspr{T,S} = 1$.
	\end{enumerate}
\end{restatable}

Finally, we bound the difference between the curvature of the uniform random walk $\curvature{T,S}$ and that of the Metropolis-Hastings (MH) random walk $\curvature{\MH; T,S}$.
Recall that this random walk proposes a move from a tree $T$ to a neighbor tree $S$ uniformly at random and then accepts the move according to the Hastings ratio, which in this case is $\min\left(1, \frac{\degree{T}}{\degree{S}}\right)$.
The mass distribution for the MH random walk thus leaves a portion of mass at the origin tree, proportional to the relative degree difference of its higher degree neighbors.
Note that the same statement and proof of Lemma~\ref{lem:asymptotic} holds with $\curvature{T,S}$ and $\ric{T,S}$ replaced by the MH curvatures $\curvature{\MH; T,S}$ and $\ric{\MH; T,S}$, respectively.

\begin{restatable}{lem}{mhcurvature}
	\label{lem:mh_curvature}
	Let $T$ and $S$ be trees with $n$ leaves. Then:
	\begin{align*}
	\curvature{T,S} - \frac{1}{3\dspr{T,S}}
	&\le \curvature{\MH; T,S} \\
	\curvature{\MH; T,S}
	&\le \curvature{T,S} + \frac{1}{3\dspr{T,S}},\text{ and }
\end{align*}
	$$\curvature{T,S} - 1/6
	\le \curvature{\MH; T,S}
	\le \curvature{T,S} + 1/6 .$$ 
\end{restatable}

\section{Conclusion and future work}
In summary, we have gone beyond graph diameter and vertex degree to substantially advance understanding of the phylogenetic rSPR graph.
We did so by developing the first theoretical and computational frameworks to bound and compute Ricci-Ollivier curvature of the rSPR graph.
We found that curvature, along with degree and distance, determine the early dynamics of hitting times for random walks.
Moreover, we proved that rSPR graph degree changes depend quadratically on the product of the size of the regrafted subtree with its change in depth, as well as that the rSPR graph tends toward flatness with respect to rSPR moves that move asymptotically small subtrees.
Finally, we proved that the coarse and asymptotic definitions of Ricci-Ollivier curvature are closely related with respect to uniform and Metropolis-Hastings walks on the rSPR graph.

In this data-free setting the stationary distribution is, unlike with real data, quite evenly spread over all trees.
Correspondingly, we found that the influence of curvature is small in this case (Fig.~\ref{fig:short-time-kappa-access}) and that the probability of the target node in the stationary distribution predominantly determines access times for pairs of trees (Fig.~\ref{fig:long-time-kappa-access}).
However, it is well known that MCMC takes a long time to approximate real phylogenetic posterior distributions even when the Bayesian credible set is small, and in fact our previous work showed significant SPR graph influence on the mixing time for phylogenetic MCMC for credible sets that had tens, hundreds or thousands of trees~\cite{Whidden2015-yi}.
Thus, our next step will be to investigate curvature of MCMC with nontrivial likelihood functions, which will reduce the posterior distribution to a more realistic effective size, and in certain cases will lead to significant ``bottlenecks'' like those we have observed in real data.
In those cases the curvature between two trees at either end of a bottleneck will describe how difficult it is to traverse the bottleneck.

Now that we have established the foundations of using curvature to understand graphs relevant for phylogenetic inference, many graph structures remain to be explored including NNI graphs, unrooted SPR graphs, graphs of ranked trees~\cite{Song2006-xe}, graphs of BEAST~\cite{Drummond2012-ek} rooted ``time-trees,'' and random walks on other discrete structures such as partitions~\cite{Gusfield2002-il} that can be expressed as trees.

\section{Acknowledgements}
The authors would like to thank Alex Gavruskin, Vladimir Minin, and Bianca Viray for helpful discussions.
They are also grateful to the authors of the SAGE and GAP4 software, especially Alexander Hulpke.

\bibliographystyle{siam}
\bibliography{curvature}

\appendix

\section{Supplementary Proofs}

\constructgraph*
\begin{proof}
	The correctness of the procedure follows by induction on the number of trees already processed, $i$, by observing that the procedure has constructed the subgraph of vertices $1, 2,  \ldots i$ and will construct the subgraph of vertices $1, 2, \ldots i+1$.

	We implement the graph with an adjacency list representation with integer-labelled vertices that supports $\OhOf{\log n}$ edge insertions and lookups (with e.g. red-black trees~\cite{guibas1978dichromatic}, as the vertex degrees are $\OhOf{n^2}$).
	As described above, the integer labels are simply the order of the input trees.
	Adding the vertices to the graph requires $\OhOf{m}$-time, as they are added in ascending order to the end of the vertex list, which can be stored as a fixed-size array.
	Adding the $\OhOf{mn^2}$ edges to the graph requires $\OhOf{mn^2 \log n}$-time.
	Enumerating the neighbors of $T_i$ requires $\OhOf{n^3}$-time for each $T_i$, for a total of $\OhOf{mn^3}$-time.
	We discuss below, in Section~\ref{sec:random_walks} how to do so efficiently without considering duplicate neighbors.
	We store the tree to index mappings for current vertices of $G$ in a trie~\cite{fredkin1960trie} using Newick representation.
	This requires only $\OhOf{n}$-time for each tree (i.e. a total of $\OhOf{mn^3}$-time) using a standard nodes-and-pointers representation of the tree and assuming integer leaf labels (a simple $\OhOf{mn\log n}$ leaf preprocessing step could be applied to extend this procedure to phylogenetic trees with string labels).
	Similarly, it takes $\OhOf{n}$-time to determine the index of each of the $\OhOf{mn^2}$ considered neighbors.
	Therefore the graph can be constructed in $\OhOf{mn^3}$-time, as claimed.
\end{proof}

\computedegree*
\begin{proof}
	The statement follows if each of the neighbor assignments are disjoint, that is $N(T,u) \cap N(T,v) = \emptyset$, for all nodes $u$, $v$ of $T$.
	So, suppose, for the purpose of obtaining a contradiction, that there exist two nodes $u$ and $v$ of $T$ such that there exists a tree $S \in (N(T,u) \cap N(T,v))$.
	Then $S$ can be obtained from $T$ by moving the subtrees rooted at $u$ or $v$.
	Call these $U$ and $V$, respectively.
	This implies that both $T \setminus U = S \setminus U$ and $T \setminus V = S \setminus V$ by the definition of an rSPR operation.
	Then the rSPR moves that move $U$ or $V$ to obtain $S$ must be nearest neighbor interchanges (NNIs), that is, rSPR moves which move their subtree to one of four locations: their grandparent edge, aunt edge, sibling's left child edge or sibling's right child edge.
	This implies that, without loss of generality, $U$ is moved to its grandparent edge and $V$ to $U's$ sibling (move type (iii)) or $U$ is moved to its aunt edge and $V$ to $U$'s edge (move type (iv)), a contradiction.
	Therefore the claim holds.
\end{proof}

\selectrandomneighbor*
\begin{proof}
	We apply the above procedure.
	We use a standard nodes-and-pointers representation of the trees, which can be constructed in $\OhOf{n}$-time from a Newick string representation and uses linear space in $n$.
	We can compute the degree of $T$ in linear time and space using Lemma~\ref{lem:compute_degree}.
	To efficiently compute $\size{N(T,u)}$ for each node $u$ of $T$, we require the number of nodes $x$ in the subtree rooted at $u$.
	We pre-compute these by (1) labeling each node with its preorder number in a preorder traversal and (2) summing the number of descendant nodes in a postorder traversal and storing the results in an array indexed by preorder number.
	Both of these traversals require $\OhOf{n}$-time.
	There are $2n-1$ = $\OhOf{n}$ nodes of $T$, and $\size{N(T,u)}$ can be computed in constant time using the subtree sizes.
	Moreover, the tree $S$ can be found in $\OhOf{n}$-time by iterating over the edges of $T$ that are not contained within $u$'s subtree to select the corresponding rSPR destination.
	Finally, we require linear time to apply the chosen rSPR operation which entails removing a node, adding a node, and updating a constant number of pointers.
	Thus, the for loop requires linear time.
	By Lemma~\ref{lem:compute_degree} the chosen tree is an rSPR neighbor of $T$ and is chosen uniformly at random.
	Therefore, the procedure uses linear time and space and selects an rSPR neighbor of $T$ uniformly at random.
\end{proof}

\degreemaxdelta*
\begin{proof}
	To prove (i), we simply note from Lemma~\ref{lem:degree_extremes} that the ladder tree achieves the minimum degree, and the balanced tree achieves the maximum degree:
	\begin{align*}
		\frac{\degree{T}}{\degree{S}} \ge\ &\frac{3n^2 - 13n + 14}{4(n-2)^2 - 2 \sum_{m=1}^{n-2} \floor{\log_2(m+1)}} \\
		\ge\ &\frac{3n^2 - 13n + 12}{4(n-2)^2 - 2(n-2)} \\
		= \ &\frac{3n^2 - 13n + 12}{4n^2 - 16n + 16 - 2(n-2)} \\
		= \ &\frac{3n^2 - 13n + 12}{4n^2 - 18n + 20} \\
		\ge\ &\frac{3n^2 - 13n + 12}{4n^2 - 17\frac{1}{3}n + 18} &\forall n \ge 3, \\
	\end{align*}
	which is greater than 3/4 when $n \ge 3$.
    Similarly for (ii):
	\begin{align*}
		{\Delta}N &= \degree{S} - \degree{T} \\
		&\le\ (4(n-2)^2 - 2 \sum_{m=1}^{n-2} \floor{\log_2(m+1)}) \\
		&\ \ \ \ \ \ \ \ - (3n^2 - 13n + 14) \\
		&\le\ (4(n-2)^2 - 2(n-2)) - (3n^2 - 13n + 14) \\
		=\ &4n^2 - 16n + 16 - 2n +4 - 3n^2 + 13n - 14 \\
		&=\ n^2 - 5n + 6.
	\end{align*}
\end{proof}

\degreechange*
\begin{proof}
The set of permissible rSPR moves changes in four different ways due to the movement of $R$:
(i) subtrees that include nodes on the path from $U$ to $L$ may now be moved into $R$ and its newly introduced parent node,
(ii) subtrees that include nodes on the path from $V$ to $L$ may no longer be moved into $R$ and its parent node,
(iii) $R$'s parent subtree may now be moved into $U$, and
(iv) $R$'s parent subtree may no longer be moved into $V$.
No additional moves are introduced or blocked by the original rSPR operation on $R$.

Recall that a rooted tree with $k$ leaves has $2(k-1)$ internal edges(recall that we are excluding any ``root edge'' in these calculations).
In the first case there are $a$ subtrees that can now be moved onto the $2k$ edges in $R$ (including its newly introduced parent edge and one of the newly subdivided root edges of $V$) for a total gain of $2ka$ distinct moves.
Similarly, we lose $2kb$ moves in the second case.
In the third case, $R$'s parent subtree may now make $2(i-1)$ moves into $U$.
Similarly, we lose $2(j-1)$ moves in the fourth case.

Thus the difference in rSPR degree is $2ka - 2kb + 2(i-1) - 2(j-1)$ as claimed.
\end{proof}

\pairedneighbors*
\begin{proof}
By the same arguments as in the proof of Lemma~\ref{lem:degree_change}, $\overlap$ rSPR moves can be applied to $T$ and $S$ with the same source and target nodes.
For each such $(T',S')$ pair, we can move $R$ in either tree to obtain the other member of the pair.
\end{proof}

\degreemaxdeltaadjacent*
\begin{proof}
	We first prove (i).
	By Lemma~\ref{lem:degree_change}, $\degree{S} - \degree{T} = 2(k(a-b) + i - j)$.
	This value is maximized by making $L$ the root and minimizing $b$, namely by setting $b=0$.
	The resulting equation $2(ka + i - j)$ is similarly maximized by setting $i=1$ (which allows us to increase $a$) then maximally balancing the terms in the product $ka$ as follows.

	There are two cases, depending on whether the subtree of $k$ leaves is moved to the root or not.
    If not, then we set $j=1$ and split the remaining $n-b-i-j = n-2$ leaves between $k$ and $a$ in as balanced a way as possible, giving (i).
	Note that this corresponds to moving the bottom subtree of $\floor{\frac{n-2}{2}}$ or $\ceil{\frac{n-2}{2}}$ leaves in a ladder tree to the root-most leaf of the tree.

	If the subtree of $k$ leaves is moved to the root, then we do not need to exclude the target branch from $k$ and $a$, gaining an additional leaf to balance the product $ka$ at the cost of increasing $j$.
	This corresponds to moving the bottom subtree of $\floor{\frac{n}{2}}$ or $\ceil{\frac{n}{2}}$ leaves in a ladder tree to the root.
	Namely, we have $2(ka + 1 - j)$, where $j = n - k = a + 1$.
	Let ${\Delta}N = \degree{S} - \degree{T}$.
	If we move the additional leaf, we have:
	\begin{align*}
		{\Delta}N &\le 2\left(\ceilXXL{\frac{n}{2}}\floorV{\frac{n-2}{2}}  + 1 - \left(\floorV{\frac{n-2}{2}} + 1\right)\right) \\
		&= 2\floorV{\frac{n-2}{2}}\ceilV{\frac{n-2}{2}},
	\end{align*}
like before.
Similarly, if we do not move the additional leaf, we also have:
\begin{align*}
	{\Delta}N &\le 2 \left(\ceilV{\frac{n-2}{2}}\floorXXL{\frac{n}{2}} +1 -  \left(\ceilV{\frac{n-2}{2}} + 1\right)\right) \\
&= 2\floorV{\frac{n-2}{2}}\ceilV{\frac{n-2}{2}},
\end{align*}
proving (i).

The relative change in degree, $\frac{\degree{T}}{\degree{S}}$, can also be written as $\frac{\degree{T}}{\degree{T} + (\degree{S} - \degree{T})}$.
By (i), we have that $\degree{S} - \degree{T} \le \frac{1}{2} (n-2)^2$,
so $\frac{\degree{T}}{\degree{S}} \ge \frac{\degree{T}}{\degree{T} + \frac{1}{2} (n-2)^2} $.
This bound is minimized when $\degree{T}$ is minimized, and recall by Lemma~\ref{lem:degree_extremes} that $\degree{T}$ is bounded below by $3n^2 - 13n + 14$.
	Thus
	\begin{align*}
		\frac{\degree{T}}{\degree{S}} &\ge \frac{3n^2 - 13n + 14}{3n^2 - 13n + 14 + \frac{1}{2}(n-2)^2} \\
		&\ge \frac{3n^2 - 13n + 14}{3.5n^2 - 15n + 16}.
	\end{align*}
	Statements (ii) and (iii) follow from this bound.

\end{proof}

\sharedneighbors*
\begin{proof}
	$T$ and $S$ differ by one rSPR move that moves a subtree $R$.
	Pick a neighbor $U \in N(T) \cap N(S)$ of both $T$ and $S$ (this intersection is not empty: $T$ and $S$ are different, so $R$ contains at most $n-2$ of the leaves, thus there must be at least one other tree $U$ obtained by moving $R$ in $T$ and $S$).
	Then either (i) $T$ and $U$ differ in the location of $R$, or (ii) $T$ and $U$ differ in the location of another subtree $Q$.
	In the latter case, $T|(X \setminus L(Q)) = S|(X \setminus L(Q))$ because $T$ and $S$ differ only in the location of $R$ and $\dspr{T,U} = \dspr{S,U} = 1$.
	Then leaves $r' \in R$, $q' \in Q$, and $u' \in U$, for some subtree $U$, form a triple of $T$ and a different triple in $S$.
	This incompatible triple can be resolved in at most $6n - 17$ ways, the maximum of which is reached when $Q$, $U$, and $R$ are themselves a ``triple'' of subtrees.
	By Lemma~\ref{lem:compute_degree}, each of the subtrees is assigned to at most $2n-6$ unique moves.
	Moreover, one additional overlapping move also moves one of the subtrees (that of the aunt of the LCA of the three subtrees).
	The number of shared neighbors is thus at most $3(2n-6) + 1 = 6n-17$.
	Note that this bound is tight when, for example, $T$ and $S$ are ladders with a different configuration of 3 leaves at maximum depth.
\end{proof}

\flatcurvature*
\begin{proof}
Because $d(T_n, S_n) = 1$, we will prove the theorem by showing that the mass transport term $W_{1,n}$ sits between two bounds, each of which has limit 1 as $n$ goes to infinity.

To start we demonstrate the theorem in the case that $T_n$ and $S_n$ have the same number of neighbors.
First we claim that $W_{1,n}$ is bounded above by $(\degree{T_n}+\OhOf{kn})/\degree{T_n}$ by exhibiting a mass transport program satisfying that bound.
Let $(T'_n, S'_n)$ be any of the $\overlap$ pairs of neighbors of $(T_n, S_n)$ which are one rSPR move apart as per Corollary~\ref{cor:paired_neighbors}.
We pair these trees in the mass transport.
There are $\OhOf{kn}$ trees unmatched by this pairing, and we can pair each of them arbitrarily with another tree of distance at most 3.
Thus, $W_{1,n}$ is bounded above by $(\degree{T_n} + \OhOf{kn}) / \degree{T_n}$.

A lower bound is also available because we can't do better than distance 1 for all trees except for shared neighbors, of which there are $\OhOf{n}$ by Lemma~\ref{lem:shared_neighbors}.
By ignoring these trees we get a lower bound of $(\degree{T_n} - (\OhOf{n}))/\degree{T_n}$ for $W_{1,n}$.

The desired control of $W_{1,n}$ is thus obtained because $\degree{T_n}$ is quadratic in $n$.

Now we prove the theorem when the number of neighbors differ.
Assume without loss of generality that $\degree{T_n} < \degree{S_n}$.
By Lemma~\ref{lem:degree_change}, $\degree{S_n} - \degree{T_n} = 2(k(a-b) + i - j)$, where each of $\{a,b,i,j\}$ is less than $n$.
Thus, $\degree{S_n} - \degree{T_n} = \OhOf{kn}$.
We again pair neighbor $T'_n$ of $T$ with neighbor $S'_n$ of $S$ such that $\dspr{T'_n, S'_n}=1$ but, as $\degree{T_n} < \degree{S_n}$ we can only account for at most $\degree{T_n} / \degree{S_n}$ of the mass directly and may have to move the $(\degree{S_n} - \degree{T_n}) / \degree{S_n}$ remainder to trees a distance at most 3.
Thus, $W_{1,n}$ is bounded above by $(\degree{T_n} + \OhOf{kn}) / \degree{S_n} = (\degree{S_n} + \OhOf{kn}) / \degree{S_n}$.
We again bound $W_{1,n}$ from below with $(\degree{T_n} - \OhOf{n}) / \degree{T_n}$ by ignoring the mass in common neighbors of $T_n$ and $S_n$.
The theorem again follows because $\degree{T_n}$ is quadratic in $n$.
\end{proof}

\curvaturedistancebound*
\begin{proof}
	Observe that the distance between neighbors of $T$ and $S$ is bounded between $\dspr{T,S} - 2$ and $\dspr{T,S} + 2$.
	For the curvature upper bound, we then have $\curvature{T,S} \le 1 - \frac{\dspr{T,S} - 2}{\dspr{T,S}} = \frac{2}{\dspr{T,S}}$.
	The lower bound follows similarly.
\end{proof}

\maxadjacentcurvature*
\begin{proof}
The maximum curvature between adjacent trees $T$ and $S$ occurs when their neighborhoods have maximum overlap and all other tree pairs are at distance 1.
By Lemma~\ref{lem:shared_neighbors} the maximum overlap is $6n-17$.
The amount of overlapping mass in the shared neighbors of $T$ and $S$ is thus $\frac{6n-17}{\max(\degree{T},\degree{S})}$.
The minimum mass transfer cost is thus $1 - \frac{6n-17}{\max(\degree{T},\degree{S})}$.
This is minimized when $\degree{T}=\degree{S}$ are as small as possible, that is $T,S$ are ladders and $\degree{T} = 3n^2 - 13n + 14$.

The maximum curvature is thus $1 - \frac{\degree{T} - (6n-17)}{\degree{T}} = \frac{6n - 17}{\degree{T}}= \frac{6n-17}{3n^2-13n+14}$.
\end{proof}

\curvaturelowerbound*
\begin{proof}
    In light of Corollary \ref{cor:paired_neighbors}, the optimal mass transport cost is maximized (and therefore curvature minimized) across adjacent trees $T$ and $S$ by a combination of two effects: trees that cannot be paired at distance $1$ and mass that must be moved between unpaired trees due to differing degrees of $T$ and $S$.
		As we will show, these effects can be optimized simultaneously.
    To bound these effects, let $m$ be the maximum (across $T$ and $S$) proportion of mass that cannot be moved between adjacent neighbors of those trees.
    We can bound the mass transport cost from above by $1 + 2m$ because pairs of neighbors of adjacent trees are at most distance 3 apart.
	This gives a lower bound of $1 - (1 + 2m) / 1) = -2m$ on the curvature.

	By Lemmas~\ref{lem:degree_change} and~\ref{lem:degree_max_delta_adjacent}, the latter effect is maximized when the relative degree change is maximized.
	By Corollary~\ref{cor:paired_neighbors}, there are at most $\overlap := \degree{T} -2ka - 2(i-1)$ paired trees, bounding the former effect.
	We now construct a pair of trees that maximizes both effects.
	Let $S$ be the ladder tree with degree $3n^2 - 13n + 14$ and $T$ be the adjacent tree constructed by moving the lower $\floor{\frac{n}{2}}$ leaves of $S$ to the root.
	$T$ has degree at most $3.5n^2 -15n + 16$.
	There are thus $2ka + 2(i-1) = 2\left(\ceil{\frac{n-2}{2}}\floor{\frac{n}{2}} +(1-1)\right) \le \frac{1}{2}n^2 - n$ unpaired neighbors, the maximum possible.
	Moreover, as shown by Lemma~\ref{lem:degree_change} this pair of trees obtains the maximum (absolute and relative) degree change.
	Thus, the maximum $m$ is:
	$$\frac{\frac{1}{2}n^2 - n}{3.5n^2 - 15n + 16}.$$
	The claim follows from multiplying this value by $-2$.
\end{proof}

\asymptotic*
\begin{proof}
	We first prove the lower bound in the uniform case, that is $\curvature{T,S} \le \ric{T,S}$.
	Let $W_1(T,S)$ be the mass transport cost in the uniform case, and $W_1'(T,S)$ be the same for the lazy uniform case with parameter $p$.
	Recall that $\curvature{T,S} = \curvature[1]{T,S} = 1 - \frac{W_1(T,S)}{\dspr{T,S}}$, and $\curvature[p]{T,S} / p = \left. \left(1 - \frac{W_1'(T,S)}{\dspr{T,S}}\right) \right/ p$.
	Observe that $$W_1'(T,S) \le pW_1(T,S) + (1-p) \, \dspr{T,S},$$ by the simple mass transport program obtained by treating the mass at $T$ and $S$ as separate from that of the neighbors.
	Then:
	\begin{align*}
		\frac{\curvature[p]{T,S}}{p} &= \left. \left(1 - \frac{W_1'(T,S)}{\dspr{T,S}}\right) \right/ p \\
		&\ge \left. \left(1 - \frac{pW_1(T,S) + (1-p)\dspr{T,S}}{\dspr{T,S}}\right) \right/ p \\
		&= \frac{1}{p} - \frac{W_1(T,S)}{\dspr{T,S}} - \frac{1-p}{p} \\
		&= 1 - \frac{W_1(T,S)}{\dspr{T,S}} \\
		&= \curvature{T,S}.
	\end{align*}

	For the upper bound, we observe that $W_1'(T,S) \ge$
	\begin{align*}
		p &W_1(T,S) + (1-p) \, \dspr{T,S} - \frac{2}{\max(\degree{T}, \degree{S})},
	\end{align*}
	as at most $1 / \max(\degree{T},\degree{S})$ of the mass can remain at each of $T$ and $S$, paired with the lazy remainder.
	The upper bound then follows analogously to the lower bound.
	Moreover, no mass can remain at $T$ or $S$ when $\dspr{T,S} > 1$, in which case the curvatures are equal.
\end{proof}

\mhcurvature*
\begin{proof}
	We first prove the lower bound.
	By Lemma~\ref{lem:degree_max_delta_adjacent}, the quotient of degrees for two adjacent trees $\ge \frac{5}{6}$.
	Thus, the Hastings ratio is always $\ge \frac{5}{6}$.
	This implies that at most $\frac{1}{6}$ of the mass remains at tree $T$ in the mass distribution.
	Let $W_1(T,S)$ be the cost of an optimal mass transport for the uniform random walk from $T$ to $S$, and $W_1'(T,S)$ the cost for the MH random walk.
	Moreover, let $m_T(z)$ and $m_S(w)$ be the mass assigned for the uniform random walk and $m'_T(z)$ and $m'_S(w)$ be the mass assigned for the MH random walk, for each vertex $z \in N(T)$ and $w \in N(S)$.
	We construct an upper bound on $W_1'(T,S)$ by moving mass according to $W_1$ where possible, and moving the remainder either from $T$ to $S$, from $T$ to a neighbor of $S$, or from a neighbor of $T$ to $S$.
	That is, for each $W_1$ assignment $\xi(z,w)$, we send $\xi'(z,w) = \xi(z,w) \min\left(\frac{m'_T(z)}{m_T(z)}, \frac{m'_S(w)}{m_S(w)}\right)$ of the mass from $z$ to $w$.
	The remaining $\xi(z,w) - \xi'(z,w)$ of the mass is moved from $T$ to $S$, $T$ to $w$, and $z$ to $S$ in the respective proportions
	$\xi(z,w) \max\left(\frac{m'_T(z)}{m_T(z)}, \frac{m'_S(w)}{m_S(w)}\right) - \xi'(z,w)$,
	$\xi(z,w) \min\left(0, \frac{m'_T(z)}{m_T(z)} - \frac{m'_S(w)}{m_S(w)}\right)$, and,
	$\xi(z,w) \min\left(0, \frac{m'_S(w)}{m_S(w)} - \frac{m'_T(w)}{m_T(w)}\right)$.
	The maximum possible mass that is not moved according to $W_1$ is $\frac{1}{6}$.
	Moreover, the affected mass must be moved through at most two additional trees.
	Then, $W_1' \le W_1 + \frac{2}{6}$.
	We now have:
	\begin{align*}
		\curvature{\MH; T,S} & \ge 1 - \frac{W_1 + \frac{1}{3}}{\dspr{T,S}} \\
		& \ge \curvature{T,S} - \frac{1}{3\dspr{T,S}}.
	\end{align*}

	In the case that $\dspr{T,S}=1$, the affected mass must be moved through only at most one additional tree, as $T$ and $S$ are adjacent.
	We thus obtain the lower bound of $\curvature{T,S} - \frac{1}{6}$ in this case.

	We obtain the upper bounds similarly to the lower bounds, by observing that the affected at most $\frac{1}{6}$ of the mass may move through at most two fewer trees (i.e. directly between $T$ and $S$ rather than a pair of neighbors at distance $\dspr{T,S} + 2$ from each other).
	Again, this is at most one fewer tree when $\dspr{T,S}=1$.
\end{proof}

\end{document}